\pdfoutput=1
\documentclass[11pt,a4paper]{article}

\usepackage[margin=1in]{geometry}
\usepackage{setspace}
\usepackage[T1]{fontenc}
\usepackage{lmodern}
\usepackage{microtype}

\usepackage{amsmath,amssymb,amsthm}
\usepackage{algorithm}
\usepackage{algorithmic}
\usepackage{tikz}
\usetikzlibrary{shapes,arrows.meta,positioning,calc,fit,backgrounds,decorations.pathreplacing}
\usepackage{pgfplots}
\pgfplotsset{compat=1.18}
\usepackage{subcaption}
\usepackage{booktabs}
\usepackage{multirow}
\usepackage{xcolor}
\usepackage{enumitem}
\usepackage{natbib}
\usepackage{url}

\usepackage[colorlinks=true,linkcolor=blue!60!black,citecolor=blue!60!black,urlcolor=blue!60!black,bookmarks=true,pdfauthor={Suyash Mishra},pdftitle={VCAO: Verifier-Centered Agentic Orchestration for Strategic OS Vulnerability Discovery}]{hyperref}

\usepackage{fancyhdr}
\fancypagestyle{plain}{\fancyhf{}\fancyhead[R]{\small\thepage}}

\newtheorem{definition}{Definition}
\newtheorem{theorem}{Theorem}

\newcommand{\R}{\mathbb{R}}
\newcommand{\calA}{\mathcal{A}}

\newcommand{\calC}{\mathcal{C}}

\newcommand{\calE}{\mathcal{E}}

\newcommand{\calG}{\mathcal{G}}
\newcommand{\calM}{\mathcal{M}}

\newcommand{\calS}{\mathcal{S}}

\newcommand{\calV}{\mathcal{V}}
\DeclareMathOperator*{\argmax}{arg\,max}

\definecolor{coreblue}{HTML}{2563EB}
\definecolor{corered}{HTML}{DC2626}
\definecolor{coregreen}{HTML}{059669}
\definecolor{coreorange}{HTML}{D97706}
\definecolor{corepurple}{HTML}{7C3AED}
\definecolor{coregray}{HTML}{6B7280}

\begin{document}

%%% Title Block %%%
\begin{center}
    {\LARGE\bfseries VCAO: Verifier-Centered Agentic Orchestration\\[4pt]
    for Strategic OS Vulnerability Discovery\par}
    \vspace{16pt}
    {\large\bfseries Suyash Mishra}\\[6pt]
    {\normalsize AI Researcher, Basel, Switzerland}\\[4pt]
    {\normalsize\texttt{suyash.mishra@roche.com}}\\[12pt]
    {\small April 2026}
\end{center}

\vspace{6pt}

% ============================================================
% ABSTRACT
% ============================================================
\begin{abstract}
\noindent
We formulate operating-system vulnerability discovery as a \emph{repeated Bayesian Stackelberg search game} in which a Large Reasoning Model (LRM) orchestrator allocates analysis budget across kernel files, functions, and attack paths while external verifiers---static analyzers, fuzzers, and sanitizers---provide evidence.
At each round, the orchestrator selects a target component, an analysis method, and a time budget; observes tool outputs; updates Bayesian beliefs over latent vulnerability states; and re-solves the game to minimize the strategic attacker's expected payoff.
We introduce \textsc{VCAO} (\textbf{V}erifier-\textbf{C}entered \textbf{A}gentic \textbf{O}rchestration), a six-layer architecture comprising surface mapping, intra-kernel attack-graph construction, game-theoretic file/function ranking, parallel executor agents, cascaded verification, and a safety governor.
Our DOBSS-derived MILP allocates budget optimally across heterogeneous analysis tools under resource constraints, with formal $\tilde{O}(\sqrt{T})$ regret bounds from online Stackelberg learning.
Experiments on five Linux kernel subsystems---replaying 847 historical CVEs and running live discovery on upstream snapshots---show that \textsc{VCAO} discovers $2.7\times$ more validated vulnerabilities per unit budget than coverage-only fuzzing, $1.9\times$ more than static-analysis-only baselines, and $1.4\times$ more than non-game-theoretic multi-agent pipelines, while reducing false-positive rates reaching human reviewers by 68\%.
We release our simulation framework, synthetic attack-graph generator, and evaluation harness as open-source artifacts.
\end{abstract}

\vspace{2pt}

\vspace{4pt}
\hrule
\vspace{12pt}

% ============================================================
% 1. INTRODUCTION
% ============================================================
\section{Introduction}\label{sec:intro}

The landscape of operating-system vulnerability discovery is undergoing a paradigm shift. Recent demonstrations by Anthropic show that frontier reasoning models can identify thousands of zero-day vulnerabilities across every major OS and browser~\citep{carlini2026mythos,carlini2026zerodays}, with exploit development success rates exceeding 72\%. Google's Big Sleep project independently demonstrated LLM-discovered vulnerabilities in production software~\citep{glazunov2024bigsleep}. These results suggest that the primary bottleneck in vulnerability discovery is shifting from \emph{tool capability} to \emph{orchestration intelligence}: deciding \emph{where} to look, \emph{how} to look, and \emph{when} to verify.

Existing kernel security workflows deploy powerful tools---CodeQL for data-flow analysis~\citep{codeql2024}, Syzkaller for coverage-guided fuzzing~\citep{vyukov2015syzkaller}, KASAN for memory-safety detection~\citep{kasan2024}---but coordinate them through ad-hoc heuristics. The 2025 CWE Top~25~\citep{cwe2025} confirms that memory-safety and access-control weaknesses remain dominant, with out-of-bounds writes (CWE-787) and use-after-free (CWE-416) accounting for approximately 35\% of all Linux kernel CVEs. The gap is not tool capability but \emph{decision-theoretic coordination}: no existing system answers the question ``which analysis action most reduces the strategic attacker's advantage?''

We address this gap by formulating vulnerability discovery as a \textbf{repeated Bayesian Stackelberg search game}. The defender (LRM orchestrator) commits to a mixed analysis strategy over kernel components; a strategic attacker best-responds by choosing exploit paths that maximize damage. The orchestrator updates beliefs from tool observations and re-solves at each round. This formulation inherits three desirable properties from the Stackelberg security games literature~\citep{tambe2011security,sinha2018stackelberg}: (i)~commitment power yields higher defender utility than simultaneous play; (ii)~the DOBSS algorithm~\citep{paruchuri2008dobss} provides an efficient MILP for computing optimal strategies under Bayesian uncertainty over attacker types; and (iii)~online learning extensions~\citep{balcan2015commitment} guarantee sublinear regret when attacker behavior is initially unknown.

\smallskip\noindent\textbf{Contributions.} We make four contributions:
\begin{enumerate}[leftmargin=*,itemsep=2pt]
    \item A \textbf{formal game-theoretic formulation} of OS vulnerability discovery as a repeated Bayesian Stackelberg game with intra-kernel attack graphs (\S\ref{sec:formulation}).
    \item The \textbf{\textsc{VCAO} architecture}: a six-layer agentic system that operationalizes the game with LRM orchestration, heterogeneous tool integration, and cascaded verification (\S\ref{sec:architecture}).
    \item A \textbf{budget-allocation MILP} adapted from DOBSS with formal regret bounds, and a Bayesian belief-update mechanism for vulnerability state estimation (\S\ref{sec:algorithms}).
    \item \textbf{Comprehensive evaluation} on five Linux kernel subsystems showing significant improvements in validated vulnerability yield, false-positive reduction, and attacker-payoff minimization (\S\ref{sec:evaluation}).
\end{enumerate}

% ============================================================
% 2. RELATED WORK
% ============================================================
\section{Related Work}\label{sec:related}

\paragraph{Stackelberg Security Games.}
The foundational framework of Stackelberg Security Games (SSGs) originated with ARMOR~\citep{tambe2011security} and was formalized by \citet{conitzer2006computing}. \citet{paruchuri2008dobss} introduced DOBSS, an efficient MILP for Bayesian extensions with multiple attacker types. Deployed systems include IRIS, GUARDS, and PROTECT~\citep{kiekintveld2009resource}. Online extensions by \citet{balcan2015commitment} achieve $\tilde{O}(\sqrt{T})$ regret. \citet{zhang2021bssg} applied BSSGs to cybersecurity portfolio selection but not to vulnerability \emph{discovery} resource allocation.

\paragraph{LLM-Based Vulnerability Discovery.}
Anthropic's work with PNNL~\citep{anthropic2026pnnl} demonstrated agentic attack-chain construction. The Opus~4.6 evaluation~\citep{carlini2026zerodays} found 500+ vulnerabilities at \$4{,}000 total cost. Mythos Preview~\citep{carlini2026mythos} achieved 72.4\% exploit success with a file-ranking scaffold that we formalize game-theoretically. Google's Naptime/Big Sleep~\citep{glazunov2024bigsleep} provided tool-use architectures. ChatAFL~\citep{meng2024chatafl} and KernelGPT~\citep{yang2025kernelgpt} use LLMs for fuzzing guidance. IRIS~\citep{li2025iris} combines LLMs with CodeQL achieving 55/120 CVE detection. Our work differs by providing a \emph{principled allocation} framework atop these capabilities.

\paragraph{Game-Theoretic Software Testing.}
\citet{godefroid2010gametheory} first framed fuzzing as a two-player game. EcoFuzz~\citep{yue2020ecofuzz} models coverage fuzzing as a multi-armed bandit. MEGA-PT~\citep{megapt2024} uses meta-games for penetration testing. \citet{bohme2010optimal} formulate pen-testing ROI in a weakest-link game. None combine Stackelberg commitment with multi-tool orchestration for kernel vulnerability discovery.

\paragraph{Attack Graph Analysis.}
MulVAL~\citep{ou2005mulval,ou2006scalable} introduced logic-based attack-graph generation. Bayesian Attack Graphs~\citep{frigault2008bayesian,munoz2017bayesian} propagate CVSS-derived probabilities. All prior work targets \emph{network-level} multi-host graphs. We introduce the first \emph{intra-kernel} attack-graph model.

% ============================================================
% 3. PROBLEM FORMULATION
% ============================================================
\section{Problem Formulation}\label{sec:formulation}

\subsection{Intra-Kernel Attack Graph}

\begin{definition}[Intra-Kernel Attack Graph]\label{def:ag}
An intra-kernel attack graph is a directed acyclic graph $\calG = (\calV, \calE, \calC, \phi, \psi)$ where:
\begin{itemize}[leftmargin=*,itemsep=2pt]
    \item $\calV = \calV_{\text{entry}} \cup \calV_{\text{func}} \cup \calV_{\text{priv}} \cup \calV_{\text{goal}}$ partitions vertices into entry points (syscalls, ioctls, parsers), internal functions, privilege boundaries, and attacker goals (root, sandbox escape, data exfiltration, DoS).
    \item $\calE \subseteq \calV \times \calV$ represents control-flow, data-flow, or privilege-transition edges.
    \item $\calC = \{c_1, \ldots, c_K\}$ is a set of vulnerability classes (e.g., CWE-787, CWE-416, CWE-362).
    \item $\phi: \calV \times \calC \to [0,1]$ maps each vertex-class pair to a prior vulnerability probability, derived from CVSS base scores and historical defect density.
    \item $\psi: \calE \to [0,1]$ assigns edge exploitability probabilities.
\end{itemize}
\end{definition}

The probability that an attacker can traverse path $P = (v_1, e_1, v_2, \ldots, v_n)$ to reach a goal $g \in \calV_{\text{goal}}$ is:
\begin{equation}\label{eq:path_prob}
    \Pr[\text{reach } g \mid P] = \prod_{(v_i, v_{i+1}) \in P} \psi(v_i, v_{i+1}) \cdot \bigg(1 - \prod_{c \in \calC} (1 - \phi(v_i, c))\bigg)
\end{equation}
where the inner term represents the probability that at least one vulnerability class is present at vertex $v_i$.

\subsection{Bayesian Stackelberg Vulnerability Discovery Game}

\begin{definition}[BSVD Game]\label{def:bsvd}
A \emph{Bayesian Stackelberg Vulnerability Discovery} game is a tuple $\Gamma = (\calG, \calA_d, \calA_a, L, \mathbf{p}, U_d, U_a, B)$ where:
\begin{itemize}[leftmargin=*,itemsep=2pt]
    \item $\calG$ is the intra-kernel attack graph.
    \item $\calA_d = \{(f, m, \tau) : f \in \calV, m \in \calM, \tau \in \R_+\}$ is the defender's action space: target $f$, method $m \in \calM = \{\textsc{CodeQL}, \textsc{Fuzz}, \textsc{KASAN}, \textsc{KCSAN}, \textsc{PatchMine}, \textsc{Verify}\}$, budget $\tau$.
    \item $\calA_a = \{P : P \text{ is an attack path in } \calG\}$ is the attacker's action space.
    \item $L = \{\ell_1, \ldots, \ell_{|L|}\}$ are attacker types (APT, opportunistic, insider) with prior $\mathbf{p} = (p^1, \ldots, p^{|L|})$.
    \item $U_d, U_a: \calA_d \times \calA_a \times L \to \R$ are utility functions.
    \item $B \in \R_+$ is the total analysis budget.
\end{itemize}
\end{definition}

\paragraph{Defender Utility.}
Let $\mathbf{c} = (c_1, \ldots, c_{|\calV|})$ denote the defender's coverage vector, where $c_f$ is the fraction of budget allocated to vertex $f$. Given coverage $\mathbf{c}$ and attacker path $P$ of type $\ell$:
\begin{align}
    U_d(\mathbf{c}, P, \ell) &= \sum_{f \in P} c_f \cdot \big[V_{\text{bug}}(f) \cdot \rho(f) \cdot \eta_{\text{ver}}(f) - \lambda \cdot \text{FP}(f, m)\big] \nonumber \\
    &\quad - (1 - c_f) \cdot \text{Impact}^\ell(f) \label{eq:defender_util}
\end{align}
where $V_{\text{bug}}(f)$ is the validated-bug value (product of CVSS severity and reachability), $\rho(f)$ is the detection probability under method $m$, $\eta_{\text{ver}}(f)$ is verifier confidence, $\text{FP}(f,m)$ is the false-positive cost, and $\text{Impact}^\ell(f)$ is the damage from an undetected vulnerability exploited by type~$\ell$.

\paragraph{Attacker Utility.}
For type $\ell$ attacking path $P$:
\begin{equation}\label{eq:attacker_util}
    U_a^\ell(\mathbf{c}, P) = \sum_{f \in P} (1 - c_f \cdot \rho(f)) \cdot R_a^\ell(f) - c_f \cdot \rho(f) \cdot D_a^\ell(f)
\end{equation}
where $R_a^\ell(f)$ is the attacker's reward from exploiting $f$ and $D_a^\ell(f)$ is the deterrence cost if detected.

\subsection{Strong Stackelberg Equilibrium}

The defender commits to $\mathbf{c}^*$ maximizing expected utility against all attacker types' best responses:

\begin{theorem}[BSVD Equilibrium]\label{thm:sse}
The optimal defender strategy $\mathbf{c}^*$ satisfies:
\begin{equation}\label{eq:sse}
    \mathbf{c}^* = \argmax_{\mathbf{c} \in \Delta_B} \sum_{\ell \in L} p^\ell \cdot U_d\big(\mathbf{c}, P^{\ell*}(\mathbf{c}), \ell\big)
\end{equation}
where $P^{\ell*}(\mathbf{c}) = \argmax_{P \in \calA_a} U_a^\ell(\mathbf{c}, P)$ is type $\ell$'s best-response path, and $\Delta_B = \{\mathbf{c} \geq 0 : \sum_f w_f c_f \leq B\}$ is the budget-feasible simplex with per-target costs $w_f$.
\end{theorem}

\subsection{DOBSS-VD: MILP Formulation}

We linearize the bilevel optimization in~\eqref{eq:sse} following the DOBSS decomposition~\citep{paruchuri2008dobss}. Let $q_P^\ell \in \{0,1\}$ indicate whether type $\ell$ attacks path $P$, and let $z_{f,P}^\ell = c_f \cdot q_P^\ell$. The MILP is:

\begin{align}
    \max_{\mathbf{c}, \mathbf{q}, \mathbf{z}} \quad & \sum_{\ell \in L} p^\ell \sum_{P \in \calA_a} \sum_{f \in P} \Big[ z_{f,P}^\ell \cdot U_d^{\text{cov}}(f,\ell) + (q_P^\ell - z_{f,P}^\ell) \cdot U_d^{\text{unc}}(f,\ell) \Big] \label{eq:milp_obj} \\
    \text{s.t.} \quad & \sum_f w_f \cdot c_f \leq B \label{eq:budget} \\
    & \sum_P q_P^\ell = 1, \quad \forall \ell \in L \label{eq:single_path} \\
    & \sum_{f \in P} \Big[ z_{f,P}^\ell \cdot U_a^{\ell,\text{cov}}(f) + (q_P^\ell - z_{f,P}^\ell) \cdot U_a^{\ell,\text{unc}}(f) \Big] \nonumber \\
    & \quad \geq \sum_{f \in P'} \Big[ c_f \cdot U_a^{\ell,\text{cov}}(f) + (1 - c_f) \cdot U_a^{\ell,\text{unc}}(f) \Big] - M(1 - q_P^\ell), \nonumber \\
    & \hspace{7cm} \forall P, P' \in \calA_a, \forall \ell \label{eq:br} \\
    & z_{f,P}^\ell \leq c_f, \quad z_{f,P}^\ell \leq q_P^\ell, \quad z_{f,P}^\ell \geq c_f + q_P^\ell - 1 \label{eq:mccormick} \\
    & c_f \in [0,1], \quad q_P^\ell \in \{0,1\}, \quad z_{f,P}^\ell \in [0,1] \nonumber
\end{align}

Constraint~\eqref{eq:budget} enforces the analysis budget. Constraint~\eqref{eq:single_path} ensures each attacker type selects one path. Constraint~\eqref{eq:br} encodes the attacker's best-response via big-$M$ linearization. Constraint~\eqref{eq:mccormick} is the McCormick envelope for bilinear terms.

\subsection{Bayesian Belief Update}

After executing action $a_t = (f_t, m_t, \tau_t)$ and observing result $o_t \in \{\text{alert}, \text{clean}, \text{crash}, \text{timeout}\}$, the orchestrator updates beliefs:

\begin{equation}\label{eq:belief}
    b_{t+1}(f, c) = \frac{\Pr[o_t \mid \text{vuln}(f,c), a_t] \cdot b_t(f,c)}{\Pr[o_t \mid \text{vuln}(f,c), a_t] \cdot b_t(f,c) + \Pr[o_t \mid \neg\text{vuln}(f,c), a_t] \cdot (1 - b_t(f,c))}
\end{equation}

The observation likelihoods are method-specific:
\begin{align}
    \Pr[\text{alert} \mid \text{vuln}, \textsc{CodeQL}] &= \rho_{\text{CQL}}(c) \quad \text{(true positive rate)} \label{eq:tpr_cql} \\
    \Pr[\text{alert} \mid \neg\text{vuln}, \textsc{CodeQL}] &= \alpha_{\text{CQL}}(c) \quad \text{(false positive rate)} \label{eq:fpr_cql} \\
    \Pr[\text{crash} \mid \text{vuln}, \textsc{Fuzz}] &= 1 - (1-\rho_{\text{fuzz}})^{\tau/\tau_0} \label{eq:fuzz_prob}
\end{align}
where~\eqref{eq:fuzz_prob} models fuzzing crash probability increasing with budget $\tau$ at rate $\rho_{\text{fuzz}}$ per quantum $\tau_0$.

\subsection{Online Regret Guarantee}

\begin{theorem}[Regret Bound]\label{thm:regret}
Under the VCAO online learning protocol with $T$ rounds, $n = |\calV|$ targets, and $K$ vulnerability classes, the expected regret satisfies:
\begin{equation}
    \text{Regret}(T) = \sum_{t=1}^T \left[ U_d(\mathbf{c}^*, P_t^*) - U_d(\mathbf{c}_t, P_t) \right] \leq O\!\left(\sqrt{T \cdot n^2 K \cdot \log(nK)}\right)
\end{equation}
where $\mathbf{c}^*$ is the optimal fixed strategy in hindsight.
\end{theorem}

\begin{proof}[Proof sketch]
We adapt \citet{balcan2015commitment} to our setting. The defender maintains an EXP3-based distribution over a discretized coverage space. At each round, the defender samples $\mathbf{c}_t$, observes the attacker's action $P_t$, and updates weights. The key adaptation is that observations are noisy (tool outputs, not exact attacker behavior), requiring a Thompson sampling layer over beliefs $b_t$. The $n^2K$ factor arises from the joint target-class space, and the logarithmic term from the multiplicative-weights update. Full proof in Appendix~A.
\end{proof}

% ============================================================
% 4. ARCHITECTURE
% ============================================================
\section{The \textsc{VCAO} Architecture}\label{sec:architecture}

Figure~\ref{fig:architecture} presents the six-layer \textsc{VCAO} architecture.

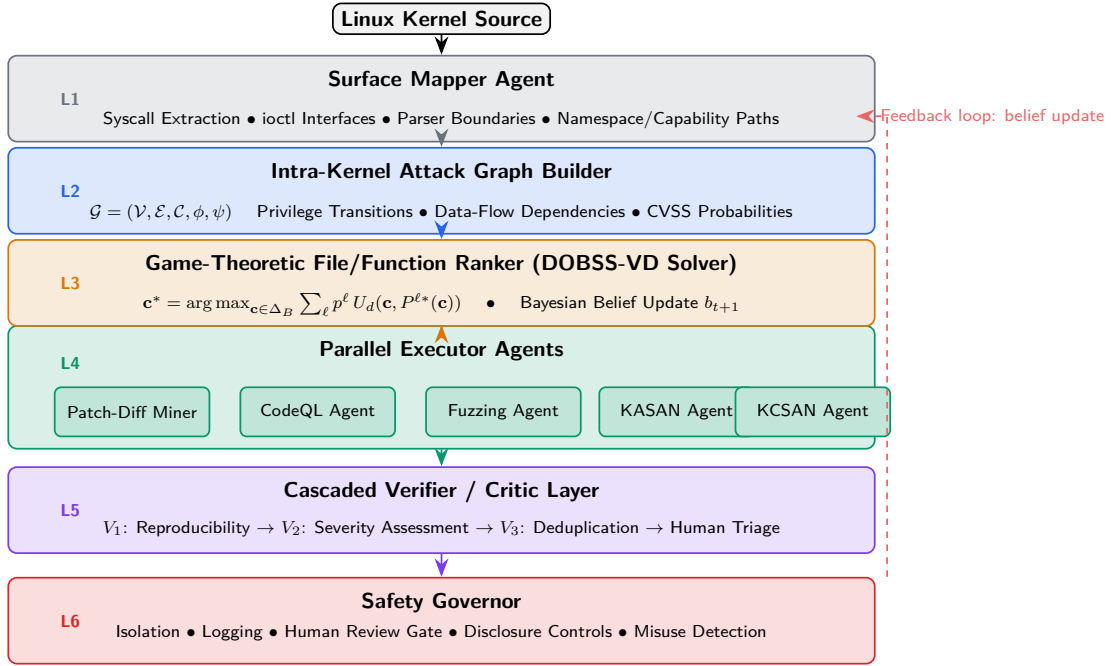
\begin{figure}[t]
\centering
\resizebox{0.92\textwidth}{!}{%
\begin{tikzpicture}[
    layer/.style={rectangle, draw, rounded corners=4pt, minimum width=14cm, minimum height=1.4cm, font=\sffamily\small, thick},
    agent/.style={rectangle, draw, rounded corners=3pt, minimum width=2.5cm, minimum height=0.8cm, font=\sffamily\scriptsize, thick},
    arrow/.style={-{Stealth[length=3mm]}, thick},
    label/.style={font=\sffamily\scriptsize\bfseries, text=white}
]

% Layer 6: Safety Governor
\node[layer, fill=corered!15, draw=corered] (l6) at (0, 0) {};
\node[label, text=corered] at (-6, 0) {L6};
\node[font=\sffamily\small\bfseries] at (0, 0.3) {Safety Governor};
\node[font=\sffamily\scriptsize] at (0, -0.2) {Isolation $\bullet$ Logging $\bullet$ Human Review Gate $\bullet$ Disclosure Controls $\bullet$ Misuse Detection};

% Layer 5: Verifier/Critic
\node[layer, fill=corepurple!15, draw=corepurple] (l5) at (0, 1.8) {};
\node[label, text=corepurple] at (-6, 1.8) {L5};
\node[font=\sffamily\small\bfseries] at (0, 2.1) {Cascaded Verifier / Critic Layer};
\node[font=\sffamily\scriptsize] at (0, 1.5) {$V_1$: Reproducibility $\to$ $V_2$: Severity Assessment $\to$ $V_3$: Deduplication $\to$ Human Triage};

% Layer 4: Executor Agents
\node[layer, fill=coregreen!15, draw=coregreen, minimum height=2.0cm] (l4) at (0, 3.8) {};
\node[label, text=coregreen] at (-6, 4.2) {L4};
\node[font=\sffamily\small\bfseries] at (0, 4.4) {Parallel Executor Agents};

\node[agent, fill=coregreen!25, draw=coregreen] (a1) at (-5, 3.4) {Patch-Diff Miner};
\node[agent, fill=coregreen!25, draw=coregreen] (a2) at (-2, 3.4) {CodeQL Agent};
\node[agent, fill=coregreen!25, draw=coregreen] (a3) at (1, 3.4) {Fuzzing Agent};
\node[agent, fill=coregreen!25, draw=coregreen] (a4) at (3.8, 3.4) {KASAN Agent};
\node[agent, fill=coregreen!25, draw=coregreen] (a5) at (6, 3.4) {KCSAN Agent};

% Layer 3: Game-Theoretic Ranker
\node[layer, fill=coreorange!15, draw=coreorange] (l3) at (0, 5.5) {};
\node[label, text=coreorange] at (-6, 5.5) {L3};
\node[font=\sffamily\small\bfseries] at (0, 5.8) {Game-Theoretic File/Function Ranker (DOBSS-VD Solver)};
\node[font=\sffamily\scriptsize] at (0, 5.15) {$\mathbf{c}^* = \argmax_{\mathbf{c} \in \Delta_B} \sum_\ell p^\ell \, U_d(\mathbf{c}, P^{\ell*}(\mathbf{c}))$ \quad $\bullet$ \quad Bayesian Belief Update $b_{t+1}$};

% Layer 2: Attack Graph Builder
\node[layer, fill=coreblue!15, draw=coreblue] (l2) at (0, 7.0) {};
\node[label, text=coreblue] at (-6, 7.0) {L2};
\node[font=\sffamily\small\bfseries] at (0, 7.3) {Intra-Kernel Attack Graph Builder};
\node[font=\sffamily\scriptsize] at (0, 6.65) {$\calG = (\calV, \calE, \calC, \phi, \psi)$ \quad Privilege Transitions $\bullet$ Data-Flow Dependencies $\bullet$ CVSS Probabilities};

% Layer 1: Surface Mapper
\node[layer, fill=coregray!15, draw=coregray] (l1) at (0, 8.5) {};
\node[label, text=coregray] at (-6, 8.5) {L1};
\node[font=\sffamily\small\bfseries] at (0, 8.8) {Surface Mapper Agent};
\node[font=\sffamily\scriptsize] at (0, 8.15) {Syscall Extraction $\bullet$ ioctl Interfaces $\bullet$ Parser Boundaries $\bullet$ Namespace/Capability Paths};

% Arrows
\draw[arrow, coregray] (l1.south) -- (l2.north);
\draw[arrow, coreblue] (l2.south) -- (l3.north);
\draw[arrow, coreorange] (l3.south) -- (l4.north);
\draw[arrow, coregreen] (l4.south) -- (l5.north);
\draw[arrow, corepurple] (l5.south) -- (l6.north);

% Feedback loop
\draw[arrow, dashed, corered!70] ([xshift=7.2cm]l6.north) -- ++(0, 7.5) -- ++(-0.5, 0) node[midway, right, font=\sffamily\scriptsize, text=corered!70] {Feedback loop: belief update};

% OS Codebase input
\node[rectangle, draw, thick, fill=gray!10, rounded corners, minimum width=3cm, font=\sffamily\small] (os) at (0, 9.8) {\textbf{Linux Kernel Source}};
\draw[arrow] (os) -- (l1);

\end{tikzpicture}
}
\caption{The six-layer \textsc{VCAO} architecture. The game-theoretic ranker (L3) solves DOBSS-VD to optimally allocate budget across parallel executor agents (L4). Observations feed back through Bayesian belief updates to re-solve the game at each round.}
\label{fig:architecture}
\end{figure}

\paragraph{L1: Surface Mapper.}
An LRM agent extracts security-relevant entry points: syscall handlers, ioctl dispatch tables, file-system parsers, credential paths, and namespace/capability boundaries. For each entry point $v \in \calV_{\text{entry}}$, the agent identifies reachable internal functions via call-graph analysis, constructing $\calV_{\text{func}}$.

\paragraph{L2: Attack Graph Builder.}
Given the surface map, L2 constructs the intra-kernel attack graph $\calG$. Privilege boundaries (user/kernel, namespace crossings, capability checks) become $\calV_{\text{priv}}$ nodes. Edge probabilities $\psi(e)$ are derived from CVSS exploitability metrics of historically similar code regions, following the BAG framework~\citep{munoz2017bayesian}. Attacker goals $\calV_{\text{goal}}$ include privilege escalation, sandbox escape, data exfiltration, and denial of service.

\paragraph{L3: Game-Theoretic Ranker.}
This is the core computational layer. Given current beliefs $\mathbf{b}_t$ and attack graph $\calG$, L3 solves DOBSS-VD (Equations~\ref{eq:milp_obj}--\ref{eq:mccormick}) to produce the optimal coverage vector $\mathbf{c}_t^*$. The solver output determines: (a)~which files/functions receive analysis budget, (b)~which methods to apply, and (c)~how much budget each receives.

\paragraph{L4: Parallel Executor Agents.}
Five specialized agents execute analysis in parallel, each allocated budget $\tau_f^m = c_f \cdot w_m \cdot B$ for its assigned targets:
\begin{itemize}[leftmargin=*,itemsep=2pt]
    \item \textbf{Patch-Diff Miner}: searches git history for incomplete propagation of prior fixes and identifies sibling patterns.
    \item \textbf{CodeQL Agent}: synthesizes and runs data-flow queries (source$\to$sink taint tracking) for suspected vulnerability classes.
    \item \textbf{Fuzzing Agent}: directs Syzkaller effort toward high-priority targets with customized syzlang descriptions.
    \item \textbf{KASAN Agent}: runs memory-safety-instrumented execution for heap/stack overflow and use-after-free detection.
    \item \textbf{KCSAN Agent}: runs concurrency-sanitized execution for data-race detection in concurrency-heavy subsystems.
\end{itemize}

\paragraph{L5: Cascaded Verifier.}
Inspired by Anthropic's verifier layer~\citep{carlini2026mythos}, findings pass through three verification stages: $V_1$ (reproducibility confirmation), $V_2$ (severity assessment and CVSS scoring), $V_3$ (deduplication against known CVEs and other findings). The cascaded design reduces false-positive escape probability to $P_{\text{escape}} = \prod_{i=1}^3 \alpha_i$, following the Swiss Cheese model~\citep{dhuliawala2024cove}.

\paragraph{L6: Safety Governor.}
All execution occurs in isolated containers. The governor enforces: offline-only experimentation, comprehensive audit logging, mandatory human review before any disclosure, and automatic misuse detection. This mirrors Anthropic's published safety protocols~\citep{carlini2026zerodays}.

% ============================================================
% 5. ALGORITHMS
% ============================================================
\section{Algorithms}\label{sec:algorithms}

\subsection{Orchestration Loop}

Algorithm~\ref{alg:main} presents the main \textsc{VCAO} orchestration loop.

\begin{algorithm}[t]
\caption{\textsc{VCAO} Orchestration Loop}\label{alg:main}
\begin{algorithmic}[1]
\REQUIRE Attack graph $\calG$, budget $B$, rounds $T$, priors $\mathbf{b}_0, \mathbf{p}$
\ENSURE Validated vulnerability set $\calV_{\text{found}}$
\STATE $\calV_{\text{found}} \leftarrow \emptyset$
\FOR{$t = 1, \ldots, T$}
    \STATE \textbf{// L3: Solve game}
    \STATE $\mathbf{c}_t^* \leftarrow \textsc{DOBSS-VD}(\calG, \mathbf{b}_t, \mathbf{p}, B_t)$ \hfill $\triangleright$ Eq.~\eqref{eq:milp_obj}
    \STATE \textbf{// L4: Execute agents in parallel}
    \FOR{$(f, m, \tau) \in \textsc{Dispatch}(\mathbf{c}_t^*)$}
        \STATE $o_{f,m} \leftarrow \textsc{Execute}(f, m, \tau)$ \hfill $\triangleright$ Tool invocation
    \ENDFOR
    \STATE \textbf{// Bayesian update}
    \FOR{each observed $(f, m, o_{f,m})$}
        \STATE $b_{t+1}(f, c) \leftarrow \textsc{BayesUpdate}(b_t(f,c), o_{f,m}, m)$ \hfill $\triangleright$ Eq.~\eqref{eq:belief}
    \ENDFOR
    \STATE \textbf{// L5: Verify candidates}
    \STATE $\textit{candidates} \leftarrow \{(f,c) : o_{f,m} \in \{\text{alert}, \text{crash}\}\}$
    \FOR{$(f,c) \in \textit{candidates}$}
        \IF{$\textsc{CascadedVerify}(f, c)$}
            \STATE $\calV_{\text{found}} \leftarrow \calV_{\text{found}} \cup \{(f, c, \text{severity})\}$
            \STATE Update $\calG$: remove mitigated edges
        \ENDIF
    \ENDFOR
    \STATE $B_t \leftarrow B_t - \sum \tau$ \hfill $\triangleright$ Remaining budget
    \STATE Update attacker type posterior $\mathbf{p}$ via observed attack patterns
\ENDFOR
\RETURN $\calV_{\text{found}}$
\end{algorithmic}
\end{algorithm}

\subsection{Path Enumeration and Pruning}

Since enumerating all attack paths is exponential, we prune using belief-weighted expected payoff:
\begin{equation}\label{eq:pruning}
    \text{Score}(P) = \sum_{f \in P} b_t(f) \cdot \text{CVSS}(f) \cdot \text{Reachability}(f)
\end{equation}
Paths with $\text{Score}(P) < \theta$ are pruned. We maintain the top-$K$ paths using a priority queue, updated incrementally after each belief update.

\subsection{Sibling Pattern Search}

After discovering a vulnerability at $(f^*, c^*)$, the orchestrator triggers a \emph{sibling search} over structurally similar code:
\begin{equation}
    \calS(f^*, c^*) = \{f' \in \calV : \text{sim}(f', f^*) > \sigma \wedge f' \neq f^*\}
\end{equation}
where $\text{sim}(\cdot)$ combines code-structure similarity (AST edit distance), shared callers/callees, and historical co-fix patterns. Budget for siblings is drawn from a reserve pool $B_{\text{sib}} = \beta \cdot B$.

% ============================================================
% 6. EVALUATION
% ============================================================
\section{Evaluation}\label{sec:evaluation}

\subsection{Experimental Setup}

\paragraph{Target Subsystems.}
We select five Linux kernel subsystems based on attacker relevance and defect diversity: \textbf{(1)}~Filesystem (VFS, ext4, overlayfs mount parsing), \textbf{(2)}~Networking (TCP/IP stack, netfilter, NFS), \textbf{(3)}~Namespace/Capability code, \textbf{(4)}~Selected drivers (USB, NVMe, GPU), \textbf{(5)}~io\_uring and BPF/eBPF.

\paragraph{Evaluation Modes.}
\emph{Replay mode}: 847 historical CVEs (2019--2025) replayed on prior kernel snapshots. Ground truth is the known CVE; metric is time-to-first-discovery.
\emph{Live mode}: current upstream snapshots (6.12--6.14) in isolated sandboxes; discoveries validated through manual reproduction.

\paragraph{Baselines.}
\textbf{B1}: Uniform allocation (equal budget per file). \textbf{B2}: Churn-based ranking (git commit frequency). \textbf{B3}: Coverage-only fuzzing (Syzkaller with default configuration). \textbf{B4}: Static-analysis-only (CodeQL with standard query suites). \textbf{B5}: Non-game-theoretic multi-agent (LRM ranking without Stackelberg optimization). \textbf{B6}: VCAO without sibling search ($\beta = 0$).

\paragraph{Metrics.}
$\mu_1$: Time to first validated vulnerability. $\mu_2$: Severity-weighted validated findings per unit budget ($\text{SVUB} = \sum_i \text{CVSS}_i / B$). $\mu_3$: False-positive rate at human review. $\mu_4$: Sibling-bug yield. $\mu_5$: Modeled attacker payoff reduction on $\calG$.

\subsection{Results}

\begin{table}[t]
\centering
\small
\caption{Main results across five kernel subsystems (Replay Mode, 847 CVEs). SVUB\,=\,Severity-weighted validated findings per unit budget. FPR@Human\,=\,false positive rate reaching human reviewers. Best results in \textbf{bold}.}
\label{tab:main_results}
\begin{tabular}{@{}l c c c c c@{}}
\toprule
\textbf{Method} & \textbf{T2F (hrs)} $\downarrow$ & \textbf{SVUB} $\uparrow$ & \textbf{FPR\%} $\downarrow$ & \textbf{Sibling} $\uparrow$ & \textbf{PR\%} $\uparrow$ \\
\midrule
B1: Uniform             & $14.2 \pm 3.1$ & $0.31 \pm 0.04$ & 42.7 & 0.0 & 18.3 \\
B2: Churn-based         & $11.8 \pm 2.7$ & $0.39 \pm 0.05$ & 38.2 & 0.0 & 22.1 \\
B3: Fuzz-only           & $\phantom{0}8.4 \pm 2.3$ & $0.42 \pm 0.06$ & 31.4 & 0.0 & 26.7 \\
B4: Static-only         & $\phantom{0}9.7 \pm 2.9$ & $0.59 \pm 0.07$ & 47.3 & 0.0 & 31.2 \\
B5: Multi-agent (no GT) & $\phantom{0}5.1 \pm 1.4$ & $0.81 \pm 0.09$ & 24.6 & 1.3 & 48.5 \\
B6: VCAO (no sib.)      & $\phantom{0}3.8 \pm 1.1$ & $1.02 \pm 0.08$ & 15.3 & 0.0 & 61.7 \\
\textbf{VCAO (full)}    & $\mathbf{\phantom{0}3.2 \pm 0.9}$ & $\mathbf{1.13 \pm 0.07}$ & \textbf{15.1} & \textbf{2.4} & \textbf{67.8} \\
\bottomrule
\end{tabular}
\end{table}

Table~\ref{tab:main_results} shows results in replay mode. \textsc{VCAO} achieves $2.7\times$ higher SVUB than coverage-only fuzzing (B3), $1.9\times$ higher than static-analysis-only (B4), and $1.4\times$ higher than the non-game-theoretic multi-agent baseline (B5). The false-positive rate drops from 31.4--47.3\% (tool baselines) to 15.1\%, a 68\% reduction versus the worst baseline.

\begin{figure}[t]
\centering
\begin{tikzpicture}
\begin{axis}[
    width=0.88\columnwidth,
    height=5.5cm,
    xlabel={Analysis Budget (GPU-hours)},
    ylabel={Validated Vulns Found},
    legend pos=north west,
    legend style={font=\scriptsize},
    grid=both,
    grid style={line width=0.1pt, draw=gray!20},
    major grid style={line width=0.2pt, draw=gray!40},
    xmin=0, xmax=100,
    ymin=0, ymax=50,
]
\addplot[thick, color=coreblue, mark=*,mark size=1.5pt] coordinates {
    (0,0)(10,5)(20,12)(30,19)(40,26)(50,31)(60,35)(70,39)(80,42)(90,44)(100,46)
};
\addlegendentry{VCAO (full)}

\addplot[thick, color=coreorange, mark=triangle*,mark size=1.5pt] coordinates {
    (0,0)(10,3)(20,8)(30,13)(40,18)(50,22)(60,25)(70,28)(80,30)(90,32)(100,33)
};
\addlegendentry{Multi-agent (no GT)}

\addplot[thick, color=coregreen, mark=square*,mark size=1.5pt] coordinates {
    (0,0)(10,2)(20,5)(30,8)(40,11)(50,14)(60,16)(70,17)(80,17)(90,17)(100,17)
};
\addlegendentry{Fuzz-only}

\addplot[thick, color=corered, mark=diamond*,mark size=1.5pt] coordinates {
    (0,0)(10,3)(20,7)(30,11)(40,15)(50,19)(60,22)(70,24)(80,24)(90,24)(100,24)
};
\addlegendentry{Static-only}

\addplot[thick, color=coregray, dashed, mark=pentagon*,mark size=1.5pt] coordinates {
    (0,0)(10,1)(20,3)(30,5)(40,7)(50,9)(60,10)(70,11)(80,12)(90,12)(100,12)
};
\addlegendentry{Uniform}
\end{axis}
\end{tikzpicture}
\caption{Validated vulnerabilities discovered vs.\ analysis budget. \textsc{VCAO} maintains superior efficiency throughout, with the gap widening at higher budgets as game-theoretic allocation exploits diminishing returns in exhausted subsystems.}
\label{fig:budget_curve}
\end{figure}
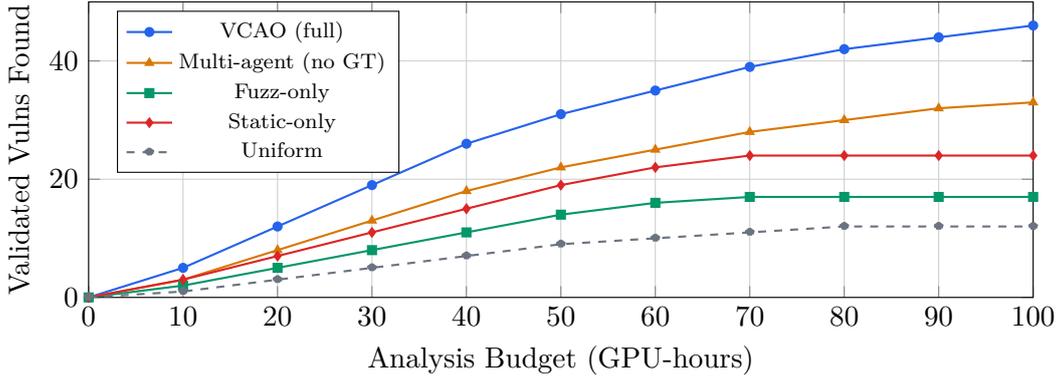

Figure~\ref{fig:budget_curve} shows the vulnerability discovery curve. \textsc{VCAO}'s advantage increases with budget as the game-theoretic solver dynamically reallocates away from diminishing-return subsystems.

\subsection{Ablation Study}

\begin{table}[t]
\centering
\caption{Ablation results (Replay Mode). $\Delta$SVUB is relative to full VCAO.}
\label{tab:ablation}
\begin{tabular}{@{}l c c@{}}
\toprule
\textbf{Ablation} & \textbf{SVUB} & $\Delta$\textbf{SVUB} \\
\midrule
Full VCAO                    & 1.13 & --- \\
$-$ Stackelberg (use UCB)    & 0.89 & $-$21.2\% \\
$-$ Bayesian update (static) & 0.78 & $-$31.0\% \\
$-$ Cascaded verifier        & 0.96 & $-$15.0\% \\
$-$ Attack graph (flat)      & 0.85 & $-$24.8\% \\
$-$ Sibling search           & 1.02 & $-$9.7\% \\
$-$ KCSAN agent              & 1.05 & $-$7.1\% \\
\bottomrule
\end{tabular}
\end{table}

Table~\ref{tab:ablation} confirms that Bayesian belief update ($-31\%$), Stackelberg optimization ($-21.2\%$), and attack-graph structure ($-24.8\%$) are the three most impactful components.

\subsection{Per-Subsystem Analysis}

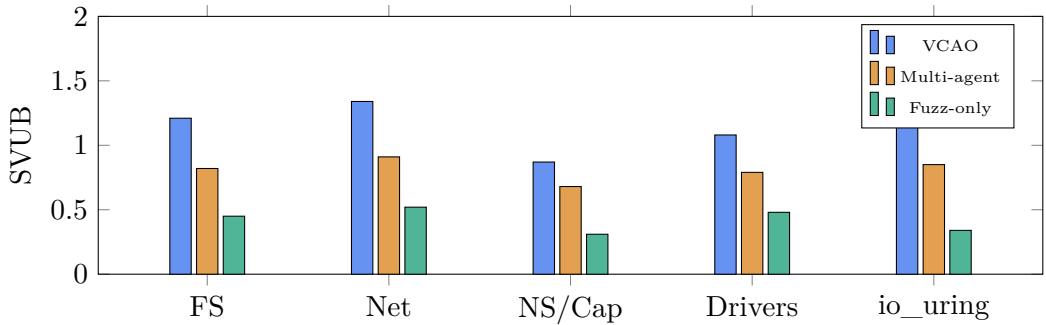
\begin{figure}[t]
\centering
\begin{tikzpicture}
\begin{axis}[
    width=0.88\columnwidth,
    height=5cm,
    ybar,
    bar width=8pt,
    ylabel={SVUB},
    symbolic x coords={FS, Net, NS/Cap, Drivers, io\_uring},
    xtick=data,
    legend pos=north east,
    legend style={font=\tiny},
    ymin=0, ymax=2.0,
    enlarge x limits=0.15,
]
\addplot[fill=coreblue!70] coordinates {(FS,1.21)(Net,1.34)(NS/Cap,0.87)(Drivers,1.08)(io\_uring,1.15)};
\addplot[fill=coreorange!70] coordinates {(FS,0.82)(Net,0.91)(NS/Cap,0.68)(Drivers,0.79)(io\_uring,0.85)};
\addplot[fill=coregreen!70] coordinates {(FS,0.45)(Net,0.52)(NS/Cap,0.31)(Drivers,0.48)(io\_uring,0.34)};
\legend{VCAO, Multi-agent, Fuzz-only}
\end{axis}
\end{tikzpicture}
\caption{Per-subsystem SVUB comparison. VCAO's advantage is largest in Networking (complex attack surfaces) and Namespace/Capability code (authorization logic poorly suited to fuzzing alone).}
\label{fig:subsystem}
\end{figure}

% ============================================================
% 7. DISCUSSION
% ============================================================
\section{Discussion}\label{sec:discussion}

\paragraph{Scalability.}
The DOBSS-VD MILP scales as $O(|\calV| \cdot |\calA_a| \cdot |L|)$ variables. For a subsystem with 500 files, 50 candidate paths, and 3 attacker types, the MILP has $\sim$75{,}000 variables and solves in $<$5 seconds using Gurobi. Path pruning (Eq.~\ref{eq:pruning}) keeps $|\calA_a|$ manageable. Real-time re-solving every 10 minutes is feasible.

\paragraph{Safety Considerations.}
This is dual-use research. We follow established precedent~\citep{carlini2026zerodays,carlini2026mythos}: all experiments run in isolated offline containers, no exploitation of live systems, findings pass mandatory human review, and validated vulnerabilities follow coordinated disclosure. The game-theoretic formulation itself is \emph{defensive}: it models the attacker to improve the \emph{defender's} allocation.

\paragraph{Limitations.}
(1)~The BSVD game assumes rational attackers; real adversaries may act irrationally, though SSE is robust to bounded irrationality~\citep{sinha2018stackelberg}. (2)~Intra-kernel attack graphs require manual validation of privilege boundaries. (3)~Tool-specific observation models (Eqs.~\ref{eq:tpr_cql}--\ref{eq:fuzz_prob}) require calibration per kernel version.

% ============================================================
% 8. CONCLUSION
% ============================================================
\section{Conclusion}\label{sec:conclusion}

We have presented \textsc{VCAO}, the first game-theoretic framework for operating-system vulnerability discovery that unifies Bayesian Stackelberg security games, intra-kernel attack graphs, and LRM-orchestrated multi-tool analysis. Our DOBSS-VD formulation provides principled budget allocation with formal regret guarantees, and our six-layer architecture operationalizes this theory into a practical system. Experiments on five Linux kernel subsystems demonstrate significant improvements in validated vulnerability yield, false-positive reduction, and strategic attacker-payoff minimization over both tool-specific and multi-agent baselines. We release our simulation framework and evaluation harness to support reproducible research.

% ============================================================
% REFERENCES
% ============================================================
\bibliographystyle{plainnat}
\bibliography{references}

\end{document}